\documentclass{amsart}
\usepackage{mathrsfs,amssymb,setspace}

\theoremstyle{plain}
\newtheorem{theorem}{Theorem}[section]

\theoremstyle{definition}
\newtheorem{notation}[theorem]{Notation}
\newtheorem{example}[theorem]{Example}
\newtheorem{definition}[theorem]{Definition}

\theoremstyle{remark}

\usepackage[all]{xy}

\newcommand{\sis}[2]{^{#1}\!\zeta_{#2}}

\begin{document}


\title[Syntactic semigroup problem for the Semigroup Reducts of $A^+(B_n)$]{Syntactic semigroup problem for the semigroup reducts of  Affine Near-semirings over\\ Brandt Semigroups}
\author[Jitender Kumar]{Jitender Kumar}
\address{School of Mathematics,  Harish-Chandra Research Institute, Allahabad, India}
\email{jitenderkumar@hri.res.in}
\author[K. V. Krishna]{K. V. Krishna}
\address{Department of Mathematics, Indian Institute of Technology Guwahati, Guwahati, India}
\email{kvk@iitg.ac.in}


\begin{abstract}
The \emph{syntactic semigroup problem} is to decide whether a given finite semigroup is syntactic or not. This work investigates the syntactic semigroup problem for both the semigroup reducts of $A^+(B_n)$, the affine near-semiring over a Brandt semigroup $B_n$. It is ascertained that both the semigroup reducts of $A^+(B_n)$ are syntactic semigroups.
\end{abstract}

\subjclass[2010]{20M35 and 68Q70.}

\keywords{Syntactic semigroup, Near-semiring, Brandt semigroup, Affine maps}

\maketitle

\section{Introduction}

The concept of syntactic semigroup (in particular, syntactic monoid) plays a vital role in algebraic automata theory. In fact, a language is recognizable if and only if its syntactic monoid is finite (cf. \cite{b.lawson,b.pin97}). Sch\"{u}tzenberger made a nontrivial use of syntactic semigroups to characterize an important subclass of recognizable languages, viz., star-free languages. Sch\"{u}tzenberger theorem states that a recognizable language is star-free if and only if its syntactic semigroup is finite and aperiodic \cite{a.schutzenberger65}. Simon established that a recognizable language language is piecewise testable if and only if its syntactic monoid is $\mathcal{J}$-trivial \cite{ic.Simon75}. Syntactic semigroups have been studied in various contexts (e.g., \cite{ic.Jackson05,a.kelarev00}).

The \emph{syntactic semigroup problem} is to decide whether a given finite semigroup is syntactic or not. Many authors have investigated the syntactic semigroup problem for various semigroups. Lallement and Milito proved that any finite semilattice of groups with identity is the syntactic monoid of a recognizable language \cite{a.Lallement75}. In \cite{a.koubek82}, Goral{\v{c}}{\'{\i}}k et al. observed that every finite inverse semigroup is syntactic. For a finite monoid $M$, an algorithm solving syntactic monoid problem for a large class of finite monoids in O($|M|^3$) time is obtained by Goral{\v{c}}{\'{\i}}k and Koubek \cite{a.koubek98}.

In \cite{a.jk13}, the authors have investigated the semigroup structure of the semigroup reducts of the affine near-semirings over Brandt semigroups.  In this paper, we study the syntactic semigroup problem for these semigroups reducts. Rest of the paper has been organized as follows. In Section 2, we present the required preliminaries. Sections 3 and 4 are devoted to the main work of the paper. We establish that the additive and also the multiplicative semigroup reducts of affine near-semirings over Brandt semigroups are syntactic semigroups in sections 3 and 4, respectively. The paper is concluded in Section 5.

\section{Preliminaries}

In this section, we present the necessary background material in two subsections. One is on syntactic semigroups to recall the essential definitions and results which are used in this paper. The other subsection is to introduce the basic structure of affine near-semirings over Brandt semigroups.

\subsection{Syntactic semigroups} The required material on syntactic semigroups is presented in this subsection. For more details one may refer to
\cite{b.lawson,b.pin97}

Let $\Sigma$  be a nonempty finite set called an \emph{alphabet} and its elements are called \emph{letters/symbols}. A \emph{word} over $\Sigma$ is a finite sequence of letters written by juxtaposing them.  The set of all words over $\Sigma$ forms a monoid with respect to concatenation of words, called the \emph{free monoid} over $\Sigma$ and it is denoted by $\Sigma^*$. The identity element of $\Sigma^*$ is the empty word (the empty sequence of letters), which is denoted by $\varepsilon$. The (free) semigroup of all nonempty words over $\Sigma$ is denoted by $\Sigma^+$. A \emph{language} over $\Sigma$ is a subset of the free monoid $\Sigma^*$.

\begin{definition}
Let $L$ be an arbitrary subset of a (multiplicative) semigroup $S$. The equivalence relation  $\approx_L$ on $S$ defined by
\[x \approx_L y   \; \text{iff} \; uxv \in L \Longleftrightarrow uyv \in L\; \text{for all}\; u, v \in S\] is a congruence known as the \emph{syntactic congruence} of $L$.
\end{definition}

\begin{definition}
For $L \subseteq \Sigma^+$, the quotient semigroup $\Sigma^+/_{\approx_L}$  is known as the \emph{syntactic semigroup} of the language $L$. Further, the quotient monoid $\Sigma^*/_{\approx_L}$ is called the \emph{syntactic monoid} of the language $L$.
\end{definition}

In the following, we present a characteristic property of syntactic semigroups through automata and their recognizable languages.

An \emph{automaton} is a quintuple $\mathcal{A} = (Q, \Sigma, q_0, T, \delta)$, where  $Q$ is a nonempty finite set called the set of \emph{states}, $\Sigma$ is an \emph{input alphabet}, $q_0 \in Q$ called the \emph{initial state} and $T$ is a subsets of $Q$, called the set of \emph{final states}, and $\delta : Q \times \Sigma \rightarrow Q$ is a function, called the \emph{transition function}.
For $x \in \Sigma^*$,  $f_x : Q \rightarrow Q$ is defined by, for all $q \in Q$,
\begin{enumerate}
\item[(i)] if $x = \varepsilon$, then $qf_x = q$;
\item[(ii)] if $x = ay$ for $a \in \Sigma$ and $y \in \Sigma^*$, then $qf_x= (q, a)\delta f_y$.
\end{enumerate}
The set of functions $\{f_x \;|\;x \in \Sigma^*\}$ forms a monoid under the composition of functions called the \emph{transition monoid} of $\mathcal{A}$, and it is denoted by $\mathscr{T}(\mathcal{A})$.

The \emph{language accepted/ recognized} by an automaton $\mathcal{A}$, denoted by $L(\mathcal{A})$, is defined by  \[L(\mathcal{A}) = \{x \in \Sigma^* \; | \; q_0f_x \in T\}.\]
A language $L$ over $\Sigma$ is said to be \emph{recognizable} if there is an automaton $\mathcal{A}$ with input alphabet $\Sigma$ such that $L(\mathcal{A}) = L$.
An  automaton $\mathcal{A}$ is said to be \emph{minimal} if the number of states of $\mathcal{A}$ is less than or equal to the number of states of any other automaton accepting $L(\mathcal{A})$.

\begin{theorem}\label{t.syn-min}
Let $L$ be a recognizable language over $\Sigma$. The syntactic monoid of $L$ is isomorphic to the transition monoid of the minimal automaton accepting $L$.
\end{theorem}

We now state another characterization for syntactic semigroups using disjunctive sets.

\begin{definition}
A subset $D$ of a semigroup $S$ is called \emph{disjunctive} in $S$ if the syntactic congruence $\approx_D$ of $D$ is the equality relation on $S$.
\end{definition}

\begin{theorem}\label{t.disj.a+bnc}
A finite semigroup is the syntactic semigroup of a recognizable language if and only if it contains a disjunctive subset.
\end{theorem}

\begin{definition}
A star-free expression over an alphabet $\Sigma$ is defined inductively as follows.
\begin{itemize}
\item [(i)] $\varnothing, \varepsilon$ and $a$, where $a \in \Sigma$, are star-free expressions representing the languages $\varnothing, \{\varepsilon\}$ and $\{a\}$, respectively.
\item [(ii)] If $r$ and $s$ are star-free expressions representing the languages $R$ and $S$, respectively, then $(r + s)$, $(rs)$ and $r^{\complement}$ are star-free expressions representing the languages $R \cup S, RS$ and $R^\complement$, respectively. Here, $R^{\complement}$ denotes the complement of $R$ in $\Sigma^*$.
\end{itemize}
A language is said to be \emph{star-free} if it can be represented by a star-free expression.
\end{definition}

\begin{theorem}[\cite{a.schutzenberger65}]\label{6.t.schu}
A recognizable language is star-free if and only if its syntactic semigroup is aperiodic.
\end{theorem}

\subsection{Affine near-semirings over Brandt semigroups}

In this subsection, we present necessary details of affine near-semirings over Brandt semigroups. For more details one may refer to \cite{t.jk14,a.jk13}

\begin{definition}
An algebraic structure $(\mathcal{N}, +, \cdot)$ is said to be a \emph{near-semiring} if
\begin{enumerate}
\item $(\mathcal{N}, +)$ is a semigroup,
\item $(\mathcal{N}, \cdot)$ is a semigroup, and
\item $a \cdot (b + c) = a \cdot b + a \cdot c$, for all $a,b,c \in \mathcal{N}$.
\end{enumerate}
\end{definition}

\begin{example}
Let $(S, +)$ be a semigroup. The algebraic structure $(M(S), +, \circ)$ is a near-semiring, where $+$ is point-wise addition and $\circ$ is composition of mappings, i.e., for $x \in S$ and $f,g \in M(S)$,
\[x(f + g)= x f + x g \;\;\;\; \text{and}\;\;\;\; x(f \circ g) = (x f)g.\]
We write an argument of a function on its left, e.g. $xf$ is the value of a function $f$ at an argument $x$. In this paper, the composition $f \circ g$ will be denoted by $fg$.
\end{example}

Let $(S, +)$ be a semigroup. An element $f \in M(S)$ is said to be an \emph{affine map} if $f = g + h$, for some endomorphism $g$ and a constant map $h$ on $S$. The set of all affine mappings over $S$, denoted by ${\rm Aff}(S)$, need not be a subnear-semiring of $M(S)$. The \emph{affine near-semiring}, denoted by $A^+(S)$, is the subnear-semiring generated by $\text{Aff}(S)$ in $M(S)$. Indeed, the subsemigroup of $(M(S), +)$ generated by $\text{Aff}(S)$ equals $(A^+(S), +)$ (cf. \cite[Corollary 1]{a.kvk05}). If $(S, +)$ is commutative, then $\text{Aff}(S)$ is a subnear-semiring of $M(S)$ so that ${\rm Aff}(S) = A^+(S)$.

\begin{definition}\label{d.bs}
For any integer $n \geq 1$, let $[n] = \{1,2,\ldots,n\}$. The semigroup
$(B_n, +)$, where $B_n = ([n]\times[n])\cup \{\vartheta\}$ and the
operation $+$ is given by
\[ (i,j) + (k,l) =
                \left\{\begin{array}{cl}
                (i,l) & \text {if $j = k$;}  \\
                \vartheta     & \text {if $j \neq k $}
                  \end{array}\right.  \]
and, for all $\alpha \in B_n$, $\alpha + \vartheta = \vartheta + \alpha = \vartheta$, is known as \emph{Brandt semigroup}. Note that $\vartheta$ is the (two sided) zero element in $B_n$.
\end{definition}

Now we recall certain concepts and results on $A^+(B_n)$ from \cite{a.jk13}, which are useful in this present work. Let $(S, +)$ be a semigroup with zero element $\vartheta$. For  $f \in M(S)$, the \emph{support of $f$}, denoted by ${\rm supp}(f)$, is defined by the set
\[ {\rm supp}(f) = \{\alpha \in S \;|\; \alpha f \neq \vartheta\}.\]
A function $f \in M(S)$ is said to be of \emph{k-support} if the
cardinality of ${\rm supp}(f)$ is $k$, i.e. $|{\rm supp}(f)| = k$. If $k = |S|$ (or $k = 1$), then $f$ is said to be of \emph{full support} (or
\emph{singleton support}, respectively). For ease of reference,  we continue to use the following notations for the elements of $M(B_n)$, as given in \cite{a.jk13}.

\begin{notation}\

\begin{enumerate}
\item For $c \in B_n$, the constant map that sends all the elements of $B_n$ to $c$ is denoted by $\xi_c$. The set of all constant maps over $B_n$ is denoted by $\mathcal{C}_{B_n}$.
\item For $k, l, p, q \in [n]$, the singleton support map that sends $(k, l)$ to $(p, q)$ is denoted by $\sis{(k, l)}{(p, q)}$.
\item For $p, q \in [n]$, the $n$-support map which sends $(i, p)$ (where $1 \le i \le n$) to $(i\sigma, q)$ using a permutation $\sigma \in S_n$ is denoted by $(p, q; \sigma)$.  We denote the identity permutation on $[n]$ by $id$.
\end{enumerate}
\end{notation}

Note that $A^+(B_1) = \{(1, 1; id)\} \cup \mathcal{C}_{B_1}$. For $n \ge 2$, the elements of $A^+(B_n)$ are given by the following theorem.

\begin{theorem}[\cite{a.jk13}]\label{2.t.ce.a+bn}
For $n \geq 2$, $A^+(B_n)$ precisely contains $(n! + 1)n^2 + n^4 + 1$ elements with the following breakup.
\begin{enumerate}
\item All the $n^2 + 1$ constant maps.
\item All the $n^4$ singleton support maps.
\item The remaining $(n!)n^2$ elements are the $n$-support maps of the form $(p, q; \sigma)$, where $p, q \in [n]$ and $\sigma \in S_n$.
\end{enumerate}
\end{theorem}

In this work, we investigate the syntactic semigroup problem for the semigroup reducts of $A^+(B_n)$. In what follows, $A^+(B_n)^{^+}$ denotes the additive  semigroup reduct $(A^+(B_n), +)$ and $A^+(B_n)^{^\circ}$ denotes the multiplicative semigroup reduct $(A^+(B_n), \circ)$ of the affine near-semiring $(A^+(B_n), +, \circ)$.

\section{The semigroup $A^+(B_n)^{^+}$}

In this section, we prove that $A^+(B_n)^{^+}$ is a syntactic semigroup of a star-free language. First we show that the monoid $A^+(B_1)^{^+}$ is syntactic by giving an isomorphism between $A^+(B_1)^{^+}$ and the syntactic monoid of a language over a set of three symbols. Then we proceed to the case for $n \ge 2$ in Theorem  \ref{6.t.a+bn+.sc}.

\begin{theorem}\label{6.t.b1+}
The monoid $A^+(B_1)^{^+}$ is a syntactic monoid of a star-free language.
\end{theorem}

\begin{proof}
Let $\Sigma = \{a, b, c\}$ and consider the language $L_b = \{x \in \{a, b\}^* \mid x \; \text{has at least one} \;  b\}$ over $\Sigma$. The language $L_b$ is recognizable. For instance, the automaton $\mathcal{A}_b$ given in Figure \ref{6.f.al1} accepts the language $L_b$.

\begin{figure}[h]
\entrymodifiers={++[o][F-]} \SelectTips{cm}{}
\[\xymatrix{*\txt{}\ar[r] & {q_0} \ar@(l, u)^a \ar@/_1pc/[dr]^c \ar@/^1.2pc/[rr]^b & *\txt{} & *++[o][F=]{q_1} \ar@(u,r)^{a, \; b } \ar@/^1pc/[dl]^c \\
*\txt{} & *\txt{} & q_2 \ar@(l,d)[]_{a, \; b, \; c} & *\txt{}}\]
\caption{The automation $\mathcal{A}_b$ accepting the language $L_b$}\label{6.f.al1}
\end{figure}
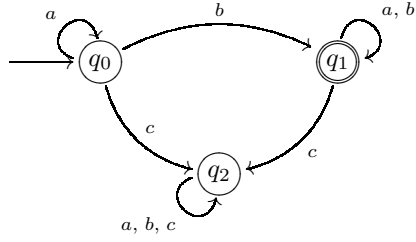

It can be easily ascertained that the automaton $\mathcal{A}_b$ is minimal (refer \cite{b.lawson} for some standard methods). By Theorem \ref{t.syn-min}, the transition monoid $\mathscr{T}(\mathcal{A}_b)$, given in Table \ref{ct+}, is syntactic.

\begin{table}[h]
\centering
\begin{tabular}{l|lll}
  & $f_a$ & $f_b$ & $f_c$\\
\hline
$f_a$ & $f_a$ & $f_b$ & $f_c$\\
$f_b$ & $f_b$ & $f_b$ & $f_c$\\
$f_c$ & $f_c$ & $f_c$ & $f_c$\\
\end{tabular}
\caption{The Cayley table for $\mathscr{T}(\mathcal{A}_b)$}\label{ct+}
\end{table}

Further, observe that the mapping $\phi : \mathscr{T}(\mathcal{A}_b) \longrightarrow A^+(B_1)^{^+}$ which assigns \[f_a \mapsto\xi_{(1, 1)} ,f_b \mapsto (1, 1; id)\;\; \text{and}\;\;  f_c \mapsto \xi_\vartheta\] is an isomorphism. Hence, the monoid $A^+(B_1)^{^+}$ is syntactic monoid. Since every element of the monoid $A^+(B_1)^{^+}$ is idempotent,  $A^+(B_1)^{^+}$ is aperiodic. Hence, by Theorem \ref{6.t.schu}, the language $L_b$ is star-free. Indeed, a star-free expression for the language $L_b$ is given by \[(\varnothing^\complement c \varnothing^\complement)^\complement b(\varnothing^\complement c \varnothing^\complement)^\complement.\]
\end{proof}

Now, for $n \ge 2$, we show that $A^+(B_n)^{^+}$ is isomorphic to the syntactic semigroup of some language over ${\rm Aff}(B_n)$.

\begin{theorem}\label{6.t.a+bn+.sc}
For $n \ge 2$, $A^+(B_n)^{^+}$ is a syntactic semigroup of a star-free language.
\end{theorem}

\begin{proof}

Let $\Sigma = {\rm Aff}(B_n)$. For $x = f_1f_2 \cdots f_k \in \Sigma^+$, write $\hat{x} = f_1 + f_2 + \cdots + f_k $. Clearly, $\hat{x} \in A^+(B_n)$. For $x, y \in \Sigma^+$, we have $\widehat{xy} = \hat{x} + \hat{y}$. Note that the function  \[\varphi : \Sigma^+ \longrightarrow A^+(B_n)^{^+}\] defined by $x \varphi = \hat x$ is an onto homomorphism.
Now for $P = \{\xi_{(1, 2)}\} \cup \left\{\!\sis{(k, l)}{(1, 1)} \; | \; k, l \in [n] \right\}$, consider the language \[L = \{x \in \Sigma^+ \; | \; \hat{x} \in P\}.\] We characterize the syntactic congruence of the language $L$ in the following claim.

\textit{Claim: For $x, y \in \Sigma^+$, $ x  \; \approx_L \; y$ if and only if $\hat{x} = \hat{y}$}.

\textit{Proof.}$(: \Leftarrow)$ Suppose $\hat x = \hat y$.  For $u, v \in \Sigma^+$,
\begin{eqnarray*}
uxv \in L  \Longleftrightarrow \widehat{uxv}  \in P  &\Longleftrightarrow& \hat{u} + \hat{x} + \hat{v} \in P  \\
     &\Longleftrightarrow&  \hat{u} + \hat{y} + \hat{v} \in P \\
     &\Longleftrightarrow& \widehat{uyv}  \in P \Longleftrightarrow  uyv \in L.
\end{eqnarray*} Hence, $x \; \approx_L \; y$.

$(\Rightarrow :)$ Suppose $\hat{x} \ne \hat{y}$. We claim that there exist $u$ and $v$ in $\Sigma^+$ such that $\hat{u} + \hat{x} + \hat{v} \in P$, whereas $\hat{u} + \hat{y} + \hat{v} \not\in P$. Consequently, $uxv \in L$ but $uyv \not\in L$ so that  $x \not\approx_L y$. In view of Theorem \ref{2.t.ce.a+bn}, we prove our claim in the following cases.

\begin{description}
  \item[{\rm\em Case 1}] $\hat{x}$ is of full support, say $\xi_{(p, q)}$. Choose $u, v \in \Sigma^+$  such that $\hat{u} = \xi_{(1, p)}$ and $\hat{v} = \xi_{(q, 2)}$. Note that  \[\hat{u} + \hat{x} + \hat{v} = \xi_{(1, p)} + \xi_{(p, q)} + \xi_{(q, 2)}  = \xi_{(1, 2)} \in P.\]
        However, as shown below, $\hat{u} + \hat{y} + \hat{v} \not\in P$, for any possibility of $\hat y \in A^+(B_n)$.
      \begin{description}
    \item[{\rm\em Subcase 1.1}] If $\hat{y} = \xi_{\vartheta}$, then clearly $\hat{u} + \hat{y} + \hat{v} = \xi_{\vartheta} \not\in P$.

    \item[{\rm\em Subcase 1.2}] Let $\hat{y}$ be of full support, say $\xi_{(r, s)}$. In either of the cases,  $p \ne r$ or $q \ne s$, we can clearly observe that\[\hat{u} + \hat{y} + \hat{v}  = \xi_{(1, p)} + \xi_{(r, s)} + \xi_{(q, 2)} = \xi_{\vartheta}\not\in P.\]

  \item[{\rm\em Subcase 1.3}]  Let $\hat{y}$ be of $n$-support, say $(k, l; \sigma)$. Let $j \in [n]$ such that $j \sigma = p$. Note that
  \[\hat{u} + \hat{y} + \hat{v} = \xi_{(1, p)} + (k, l; \sigma) +  \xi_{(q, 2)}=
                \left\{\begin{array}{cl}
                \xi_{\vartheta} & \text {if $l \ne q$;}  \\
                \sis{(j, k)}{(1, 2)}   & \text {if $l = q$.}
                  \end{array}\right. \]

  \item[{\rm\em Subcase 1.4}] Let $\hat{y}$ be of singleton support, say $\sis{(k, l)}{(r, s)}$. Note that
\[\hat{u} + \hat{y} + \hat{v} = \xi_{(1, p)} + \;\sis{(k, l)}{(r, s)} +  \xi_{(q, 2)} =
                \left\{\begin{array}{cl}
                \xi_{\vartheta} & \text {if $p \ne r$ or $q \ne s$;} \\
                \sis{(k, l)}{(1, 2)} & \text {otherwise.}

                  \end{array}\right. \]

\end{description}
  \item[{\rm\em Case 2}] $\hat{x}$ is of $n$-support, say $(p,q; \sigma)$.
  \begin{description}

  \item[{\rm\em Subcase 2.1}]  Let $\hat{y}$ be of $n$-support, say $(k, l; \tau)$. Choose $u, v \in \Sigma^+$ such that $\hat{u} = \;\sis{(l, p)}{(1, l\sigma)}$ and $\hat{v} = \xi_{(q, 1)}$. Note that \[\hat{u} + \hat x + \hat{v} =  \;\sis{(l, p)}{(1, l\sigma)} + (p,q; \sigma) + \xi_{(q, 1)} = \;\sis{(l, p)}{(1, 1)} \in P.\]
      If $p \ne k$ or $q \ne l$, we have
   \[\hat{u} + \hat{y} + \hat{v} = \;\sis{(l, p)}{(1, l\sigma)} + (k, l; \tau) + \xi_{(q, 1)} = \xi_{\vartheta} \not\in P.\] Otherwise, we have $\sigma \ne \tau$. there exists $j_0 \in [n]$ such that $j_0 \sigma \ne j_0 \tau$. Now choose $u, v \in \Sigma^+$ such that $\hat{u} = \;\sis{(j_0, p)}{(1, j_0 \sigma)}$ and $\hat{v} = \xi_{(q, 1)}$. Note that $\hat{u} + \hat{x} + \hat{v}= \;\sis{(j_0, p)}{(1, 1)} \in P$, whereas \[\hat{u} + \hat{y} + \hat{v} = \;\sis{(j_0, p)}{(1, j_0 \sigma)} + (p,q; \tau) + \xi_{(q, 1)} = \xi_{\vartheta} \not\in P.\]

 \item[{\rm\em Subcase 2.2}] Let $\hat{y}$ be of singleton support, say $\sis{(j, k)}{(m, r)}$. Choose $u, v \in \Sigma^+$ such that $\hat{u} = \;\sis{(l, p)}{(1, l\sigma)}$ (with $l\sigma \ne m$) and $\hat{v} = \xi_{(q, 1)}$. Note that \[\hat{u} + \hat{x} + \hat{v} = \;\sis{(l, p)}{(1, 1)}\in P,\] whereas  \[\hat{u} + \hat{y} + \hat{v} = \;\sis{(l, p)}{(1, l\sigma)}+ \;\sis{(j, k)}{(m, r)} + \xi_{(q, 1)} = \xi_{\vartheta} \not\in P.\]

\item[{\rm\em Subcase 2.3}] If $\hat{y} = \xi_{\vartheta}$, then, for $\hat{u} = \;\sis{(l, p)}{(1, l\sigma)}$ and $\hat{v} = \xi_{(q, 1)}$  we have $\hat{u} + \hat x + \hat{v} = \;\sis{(l, p)}{(1, 1)} \in P$, but $\hat{u} + \hat{y} + \hat{v} = \xi_{\vartheta} \not\in P$.
 \end{description}

\item[{\rm\em Case 3}] $\hat{x}$ is of singleton support, say $\sis{(p, q)}{(r, s)}$. Choose $u, v \in \Sigma^+$ such that $\hat{u} = \;\sis{(p, q)}{(1, r)}$ and $\hat{v} = \;\sis{(p, q)}{(s, 1)}$ so that
      \[\hat{u} + \hat{x} + \hat{v} = \;\sis{(p, q)}{(1, r)} +  \; \sis{(p, q)}{(r, s)}+ \;\sis{(p, q)}{(s, 1)} = \;\sis{(p, q)}{(1, 1)}\in P.\]

\begin{description}
\item[{\rm\em Subcase 3.1}] $\hat{y}$ is of singleton support, say $\sis{(j, k)}{(l, m)}$. Since $\hat{x} \ne \hat{y}$, we have \[\hat{u} + \hat{y} + \hat{v} = \;\sis{(p, q)}{(1, r)} +  \;\sis{(j, k)}{(l, m)} + \;\sis{(p, q)}{(s, 1)}= \xi_{\vartheta}\not\in P.\]

\item[{\rm\em Subcase 3.2}] If $\hat{y} = \xi_{\vartheta}$, then clearly $\hat{u} + \hat{y} + \hat{v} = \xi_{\vartheta} \not\in P$.
\end{description}
\end{description}

Thus, we have proved that, for $x, y \in \Sigma^+$, $ x  \; \approx_L \; y$ if and only if $\hat{x} = \hat{y}$. Consequently, $\ker \varphi = \; \approx_L$. By fundamental theorem of homomorphisms of semigroups, we have $\Sigma^+/_{\approx_L}$ is isomorphic to $A^+(B_n)^{^+}$. Hence, $A^+(B_n)^{^+}$ is a syntactic semigroup of the language $L$. However, since the semigroup $A^+(B_n)^{^+}$ is aperiodic (cf. \cite[Proposition 3.10]{a.jk13}), by Theorem \ref{6.t.schu}, the language $L$ is star-free.
\end{proof}

\section{The semigroup $A^+(B_n)^{^\circ}$}

In order to ascertain that the semigroup $A^+(B_n)^{^\circ}$ is syntactic, in this section, first we show that the monoid $A^+(B_1)^{^\circ}$ is syntactic. Then, for $n \ge 2$, we show that the semigroup $A^+(B_n)^{^\circ}$ is syntactic by constructing a disjunctive subset.

\begin{theorem}\label{6.t.b1c}
The monoid $A^+(B_1)^{^\circ}$ is a syntactic monoid of a star-free language.
\end{theorem}

\begin{proof}
Let $\Sigma = \{a, b, c\}$ and consider the language $L_a = \{xab^n \mid x \in \Sigma^* \; \text{and} \;  n \ge 0\}$ over $\Sigma$. The language $L_a$ is recognizable. For instance, the automaton $\mathcal{A}_a$ given in Figure \ref{6.f.al1c} is a minimal automaton accepting the language $L_a$. By Theorem \ref{t.syn-min}, the transition monoid $\mathscr{T}(\mathcal{A}_a)$, given in Table \ref{ctc}, is a syntactic monoid.

\begin{figure}[htbp]
\entrymodifiers={++[o][F-]} \SelectTips{cm}{}
\[\xymatrix{ *\txt{} \ar[r] &  q_0 \ar@(l,u)^{b, \; c} \ar@/^1pc/[rr]^a & *\txt{} & *++[o][F=]{q_1} \ar@(u,r)^{a, \; b } \ar@/^1pc/[ll]^c } \]

\caption{The automation $\mathcal{A}_a$ accepting the language $L_a$}\label{6.f.al1c}
\end{figure}
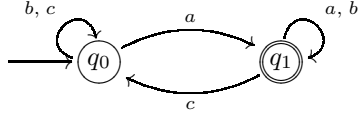

\begin{table}[h]
\centering

\begin{tabular}{l|lll}
  & $f_a$ & $f_b$ & $f_c$\\
\hline
$f_a$ & $f_a$ & $f_a$ & $f_c$\\
$f_b$ & $f_a$ & $f_b$ & $f_c$\\
$f_c$ & $f_a$ & $f_c$ & $f_c$\\
\end{tabular}
\caption{The Cayley table for $\mathscr{T}(\mathcal{A}_a)$}\label{ctc}
\end{table}

Further, it is easy to check that the mapping $\phi : \mathscr{T}(\mathcal{A}_a) \longrightarrow A^+(B_1)^{^\circ}$ which assigns $f_a \mapsto\xi_{(1, 1)} ,f_b \mapsto (1, 1; id)$ and $f_c \mapsto \xi_\vartheta$ is an isomorphism. Hence, the monoid $A^+(B_1)^{^\circ}$ is a syntactic monoid.
Since every element of the semigroup $A^+(B_1)^{^\circ}$ is idempotent, $A^+(B_1)^{^\circ}$ is aperiodic so that the language $L_a$ is star-free. A star-free expression for the language $L_a$ is given by \[\varnothing^\complement a (\varnothing^\complement(a + c)\varnothing^\complement)^\complement.\]
\end{proof}

\begin{theorem}\label{6.t.smbnc}
The set $D = \{(1, 1; id)\} \cup \{\xi_{(p, q)}\; | \; p, q \in [n]\}$ is a disjunctive subset of the semigroup $A^+(B_n)^{^\circ}$.
Hence, $A^+(B_n)^{^\circ}$ is a syntactic semigroup.
\end{theorem}

\begin{proof}
Let $f, g \in A^+(B_n)^{^\circ}$ such that $f \ne g$. We claim that there exist $h$ and $h'$ in $A^+(B_n)^{^\circ}$ such that only one among $hfh'$ and $hgh'$ is in $D$ so that $f \not\approx_D g$. Since $f$ and $g$ are arbitrary, it follows that $\approx_D$ is the equality relation on $A^+(B_n)^{^\circ}$. We prove our claim in various cases on supports of $f$ and $g$ (cf. Theorem \ref{2.t.ce.a+bn}).

\begin{description}
\item[{\rm\em Case 1}] $f$ is of $n$-support, say $(p, q; \sigma)$. Choose $h = (1, p; id)$ and $h' = (q, 1; \sigma^{-1})$. Note that
  \[hfh' = (1, p; id)(p, q; \sigma)(q, 1; \sigma^{-1}) = (1, 1; id) \in D.\] If $g = \xi_{\vartheta}$, then clearly $hgh' = \xi_{\vartheta} \not\in D$. If $g$ is of $n$-support, say $(k, l; \tau)$, then
\[hgh' = (1, p; id)(k, l; \tau)(q, 1; \sigma^{-1}) = \left\{\begin{array}{cl}
                                              \xi_{\vartheta} & \text {if $p \ne k$ or $q \ne l$} \\
                                              (1, 1; \tau\sigma^{-1}) & \text{otherwise}
                                              \end{array}\right. \] so that $hgh' \not\in D$.
In case $g$ is of singleton support, say $\sis{(r , s)}{(u, v)}$,
     \[hgh' = (1, p; id)\sis{(r, s)}{(u, v)}(q, 1; \sigma^{-1}) \; \text{is either $\xi_{\vartheta}$ or is of singleton support}\] so that  $hgh' \not\in D.$ Finally, let $g$ be a full support element, say $g = \xi_{(k, l)}$. Now, for $h = \xi_{\vartheta}$ and $h' = \;\sis{(k, l)}{(s, t)}$, we have
\[hgh' = \xi_{\vartheta}\xi_{(k, l)}\;\sis{(k, l)}{(s, t)} = \xi_{(s, t)} \in D\] whereas
\[hfh' = \xi_{\vartheta}(p, q; \sigma)\sis{(k, l)}{(s, t)} = \xi_{\vartheta} \not\in D.\]

\item[{\rm\em Case 2}] $f$ is of full support, say $\xi_{(p, q)}$. If $g$ is a constant map, then choose $h = \xi_{(p, q)}$ and $h' = \;\sis{(p, q)}{(u, v)}$ so that
\[hfh' = \xi_{(p, q)}\xi_{(p, q)}\;\sis{(p, q)}{(u, v)} = \xi_{(u, v)}\in D\]
but \[hgh' = \xi_{(p, q)}g\;\sis{(p, q)}{(u, v)} = \xi_\vartheta \not\in D.\]

In case $g$ is of singleton support, say $\sis{(r, s)}{(u, v)}$,  choose $h = \xi_\vartheta$ and $h' = (q, q; id)$. Note that
\[hfh'  = \xi_\vartheta\xi_{(p, q)}(q, q; id) = \xi_{(p, q)} \in D\] whereas
\[hgh' = \xi_\vartheta \; \sis{(r, s)}{(u, v)}(q, q; id) = \xi_\vartheta \not\in D.\]

\item[{\rm\em Case 3}] $f$ is of singleton support, say $\sis{(p, q)}{(r, s)}$. Now, we will only assume $g$ is a singleton support map, say $\sis{(k, l)}{(u, v)}$, or the zero map. In any case, for $h = \xi_{(p, q)}$ and $h' = \;\sis{(r, s)}{(u, v)}$, we have $hgh' = \xi_{(p, q)}g\sis{(r, s)}{(u, v)} = \xi_{\vartheta}\not\in D$, whereas \[hfh' = \xi_{(p, q)}\;\sis{(p, q)}{(r, s)}\;\sis{(r, s)}{(u, v)} = \xi_{(u, v)} \in D.\]

\end{description}
\end{proof}

\section{Conclusion}

In this work, we have shown that both the semigroup reducts of $A^+(B_n)$ are syntactic semigroups by deploying various techniques, viz. minimal automata and disjunctive subsets. We established that $A^+(B_1)^{^+}$ and $A^+(B_1)^{^\circ}$ are the transition monoids of some minimal automata. For $n \ge 2$, we proved that $A^+(B_n)^{^+}$ is isomorphic to the syntactic semigroup of a star-free language. Further, for $n \ge 2$, we have shown that the semigroup $A^+(B_n)^{^\circ}$ contains a disjunctive subset.


\end{document}